\documentclass[11pt]{article}

\usepackage{fullpage}
\usepackage{here}
\usepackage{amsthm,amsmath,amssymb}
\usepackage{amsfonts}
\usepackage{yhmath}
\usepackage{xcolor}
\usepackage{graphicx}
\usepackage{booktabs}
\usepackage{caption}
\usepackage{hyperref}
\usepackage{enumerate}
\usepackage{enumitem}
\usepackage{wrapfig}

\theoremstyle{plain}
\newtheorem{theorem}{Theorem}
\newtheorem{lemma}[theorem]{Lemma}

\begin{document}
\title{On RAC Drawings of Graphs with Two Bends per Edge\thanks{Research on this paper was partially supported by the NSF award DMS~2154347.}}
\author{Csaba D. T\'oth\thanks{California State University Northridge, Los Angeles, CA, USA; and
Tufts University, Medford, MA, USA. Email: \texttt{csaba.toth@csun.edu}}}
\date{}

\maketitle

\begin{abstract}
It is shown that every $n$-vertex graph that admits a 2-bend RAC drawing in the plane, where the edges are polylines with two bends per edge and any pair of edges can only cross at a right angle, has at most $20n-24$ edges for $n\geq 3$. This improves upon the previous upper bound of $74.2n$; this is the first improvement in more than 12 years. A crucial ingredient of the proof is an upper bound on the size of plane multigraphs with polyline edges in which the first and last segments are either parallel or orthogonal.
\end{abstract}

\section{Introduction}
\label{sec:intro}

\textbf{Right-angle-crossing  drawings} (for short, \textbf{RAC drawings}) were introduced by Didimo et al.~\cite{DidimoEL11}. In a RAC drawing of a graph $G=(V,E)$, the vertices are distinct points in the plane, edges are polylines, each composed of finitely many line segments, and any two edges can cross only at a $90^\circ$ angle. For an integer $b\geq 0$, a \textbf{RAC$_b$ drawing} is a RAC drawing in which every edge is a polyline with at most $b$ \emph{bends}; and a \textbf{RAC$_b$ graph} is an abstract graph that admits such a drawing. Didimo et al.~\cite{DidimoEL11} proved that every RAC$_0$ graph on $n\geq 4$ vertices has at most $4n-10$ edges, and this bound is tight when $n=3h-5$ for all $h\geq 3$; see also \cite{DujmovicGMW11}. They also showed that every graph is a RAC$_3$ graph.

Angelini et al.~\cite{AngeliniBFK20} proved that every RAC$_1$ graph on $n$ vertices has at most $5.5n-O(1)$ edges, and this bound is the best possible up to an additive constant.
The only previous bound on the size of RAC$_2$ graphs is due to Arikushi et al.~\cite{ArikushiFKMT12}: They showed that every RAC$_2$ graph on $n$ vertices has at most $74.2n$ edges, and constructed RAC$_2$ graphs with $\frac{47}{6}n-O(\sqrt{n})>7.83n-O(\sqrt{n})$ edges.
Recently, Angelini et al.~\cite{AgeliniBK00U23} constructed an infinite family of RAC$_2$ graphs with $10n-O(1)$ edges;
and conjectured that this lower bound is the best possible.

See recent surveys~\cite{Didimo20,DidimoLM19} and results~\cite{AngeliniBKKP22,Forster020,RahmatiE20,Schaefer21} for other aspects of RAC drawings. The concept of RAC drawings was also generalized to angles other than $90^\circ$, and to combinatorial constraints on the crossing patterns in a drawing~\cite{AckermanFT12}.

The main result of this note is the following theorem.

\begin{theorem}\label{thm:main}
  Every RAC$_2$ graph with $n\geq 3$ vertices has at most $20n-24$ edges.
\end{theorem}

This improves upon the upper bound $24n-26$ in the conference version of this paper~\cite{Toth23},
which in turn was the first improvement on the size of RAC$_2$ graphs in more than 12 years.

\paragraph{Related Results and Open Problems.}
Several special cases of the problem have also been considered:  A drawing of a graph is \emph{simple} if any pair of edges share at most one point, which may be a crossing or a common endpoint. In a non-simple drawing, a \emph{lens} is a region bounded by a closed Jordan curve, comprised of two Jordan arcs, each of which is part of the drawing of an edge. A drawing of a graph is \emph{non-homotopic} if the interior of every lens contains a vertex or a crossing. Note that every simple drawing is non-homotopic (since it does not contain any lens). In the \emph{general} case (e.g., the setting of Theorem~\ref{thm:main}), there are no such restrictions on the drawings.

Recently, Kaufmann et al.~\cite{kaufmann2023density} proved an upper bound of $10n-19$ for the number of edges in non-homotopic RAC$_2$ drawings with $n\geq 3$ vertices. They also constructed a simple RAC$_2$ drawing with $10n-O(1)$ edges and $n=k^2+8$ vertices for all $k\geq 1$ (the earlier lower bound construction for $10n-O(1)$ by Angelini et al.~\cite{AgeliniBK00U23} was neither simple nor non-homotopic). Thus the bound $10n-O(1)$ is tight for non-homotopic and for simple RAC$_2$ drawings.
It is also known that every simple (resp., non-homotopic) RAC$_1$ drawing with $n$ vertices has at most $5n-O(1)$ edges, and this bound is the best possible~\cite{AngeliniBFK20,kaufmann2023density}.

Schaefer~\cite{Schaefer21} proved that recognizing RAC$_0$ graphs is $\exists\mathbb{R}$-complete (this problem was previously known to be NP-hard~\cite{ArgyriouBS12}). It is also $\exists\mathbb{R}$-complete to decide whether a graph admits a RAC$_0$ drawing isomorphic to a given drawing in which every edge has at most eleven crossings.
It is an open problem whether RAC$_1$ and RAC$_2$ graphs can be recognized efficiently; this problem is open even
if all crossing edge pairs are given.

\section{Multigraphs with Angle-Constrained End Segments}
\label{sec:tech}

A plane multigraph $G=(V,E)$ is a multigraph embedded in the plane such that the vertices are distinct points, and the edges are Jordan arcs between the corresponding vertices (not passing through any other vertex), and any pair of edges may intersect only at vertices. The \emph{multiplicity} of an edge  between vertices $u$ and $v$ is the total number of edges in $E$ between $u$ and $v$.

We define a \textbf{plane ortho-fin multigraph} as a plane multigraph $G=(V,E)$ such that every edge $e\in E$ is a polygonal path $e=(p_0,p_1,\ldots , p_k)$ where the first and last edge segments are either parallel or orthogonal, that is, $p_0p_1\| p_{k-1}p_k$ or $p_0p_1\perp p_{k-1}p_k$. See Fig.~\ref{fig:1} for  examples.

\begin{figure}[htbp]
\centering
\includegraphics[width=0.8\textwidth]{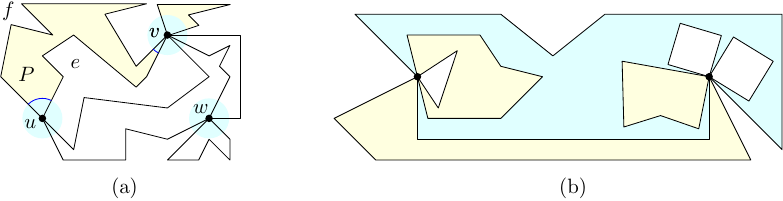}
\caption{(a) A plane ortho-fin multigraph with 3 vertices and 7 edges.
The parallel edges $e$ and $f$ form a simple polygon $P$ with potential $\Phi(P)=\pi/2$.
(b) A plane ortho-fin multigraph with 2 vertices and 7 edges.} \label{fig:1}
\end{figure}

It is not difficult to see that in a plane ortho-fin multigraph, every vertex is incident to at most three loops and the multiplicity of any edge between two distinct vertices is at most eight; both bounds can be attained (see Fig.~\ref{fig:tight} below for examples). Combined with Euler's formula, this would already give an upper bound of $3n+8(3n-6)=27n-48$ for the size of a plane ortho-fin multigraph with $n\geq 3$ vertices. In this section, we prove a tight bound of $5n-2$ (Theorem~\ref{thm:orthogonal})\footnote{The conference version of this paper~\cite{Toth23} established
a weaker upper bound of $7n-3$.}.
The key technical tool is the following.
For a face $P$ of a plane multigraph $G=(V,E)$, let the \textbf{potential} $\Phi(P)$ be the sum of interior angles of $P$ over all vertices in $V$ incident to $P$.

\begin{lemma}\label{lem:angle}
Let $G=(V,E)$ be a plane ortho-fin multigraph. Then for every face $P$, the potential $\Phi(P)$ is a multiple of $\pi/2$.
\end{lemma}
\begin{proof}
Let $G=(V,E)$ be a plane ortho-fin multigraph with $n$ vertices. We may assume w.l.o.g.\ that $G$ is connected.
Since $\Phi(P)$ is determined by the edges and vertices of $G$ on the boundary of $P$,
we may assume that all edges and vertices of $G$ are incident to $P$.

\paragraph{Cycle Multigraphs.}
Assume first that $G$ is a cycle (possibly a loop or a double edge); this assumption is dropped later.
We may further assume, w.l.o.g., that $P$ lies in the interior of the cycle $G$. Indeed, denote the interior and exterior face of $G$ by $P_{\rm  int}$ and $P_{\rm ext}$, respectively. At each vertex of $G$, the interior and exterior angles sum to $2\pi$. Consequently, $\Phi(P_{\rm int})+\Phi(P_{\rm ext})=2\pi\cdot n$. It follows that if $\Phi(P_{\rm int})$ is a multiple of $\pi/2$, then so is $\Phi(P_{\rm ext})$.

We distinguish between three cases based on the number of vertices in $V$.

\noindent (1) Assume that $G$ has only one vertex, denoted $v\in V$. Then $P$ is bounded by a counterclockwise loop $e=(p_0,p_1,\ldots , p_k)$ incident to the vertex $v=p_0=p_k$. Since $p_0p_1\| p_{k-1}p_k$ or $p_0p_1\perp p_{k-1}p_k$, then the interior angle of $P$ at $v$ is $\pi/2$, $\pi$, or $3\pi/2$. We see that $\Phi(P)$ is a multiple of $\pi/2$.

\noindent (2) Assume that $G$ has two vertices, $u,v\in V$.
Then $P$ is bounded by parallel edges $e=(p_0,p_1,\ldots , p_k)$ and $f=(q_0,q_1,\ldots , q_\ell)$ between $u=p_0=q_0$ and $v=p_k=q_\ell$. Assume w.l.o.g.\ that $e$ is oriented counterclockwise along $P$; consequently, $f$ is oriented clockwise. The interior angles of $P$ at $u$ and $v$ are $\angle p_1uq_1$ and $\angle q_{\ell} v p_k$. Note, in particular, that $\Phi(P)=\angle p_1uq_1+\angle q_{\ell} v p_k$, and $\Phi(P)$ depends only on the directions of the vectors $\overrightarrow{up_1}$, $\overrightarrow{uq_1}$, $\overrightarrow{vp_k}$, and $\overrightarrow{v q_{\ell}}$. If $\overrightarrow{up_1}$ and $\overrightarrow{vp_k}$ have the same direction, and so do $\overrightarrow{u q_1}$ and $\overrightarrow{v q_{\ell}}$, then $\angle p_1uq_1+\angle q_{\ell} v p_k=\angle p_1uq_1+\angle q_1 u p_1=2\pi$. In general, the directions of $\overrightarrow{up_1}$ and $\overrightarrow{vp_k}$ (resp., $\overrightarrow{u q_1}$ and $\overrightarrow{v q_{\ell}}$) differ by a multiple of $\pi/2$. Consequently, $\Phi(P)=\angle p_1uq_1+\angle q_{\ell} v p_k$ is also a multiple of $\pi/2$.

\noindent (3) Let $P$ be bounded by a simple closed curve $\gamma$ that passes through $k\geq 3$ vertices, $v_1,\ldots , v_k\in V$, in counterclockwise order.
Let $\gamma'=(v_1,\ldots , v_k)$ be a (not necessarily simple) polygonal curve, with straight-line edges between consecutive vertices in $V$. Then the sum of angles on the left side of $\gamma'$ at the vertices is $V$ is a multiple of $\pi$. We can transform $\gamma'$ to $\gamma$ by successively replacing the straight-line edges $v_iv_{i+1}$ with the corresponding polyline edges of the ortho-fin multigraph $G$. If the first and last edge of the ortho-fin edge between $v_i$ and $v_{i+1}$ have the same direction, then replacing the straight-line edge with such an ortho-fin edge does not change the sum of interior angles. In any other case, the sum of interior angles changes by a multiple of $\pi/2$.
Consequently, the sum of angles over vertices in $V$ in the polygon in the interior of $\gamma$ is also a multiple of $\pi/2$. This proves that $\Phi(P)$ is a multiple of $\pi/2$.

\begin{figure}[htbp]
\centering
\includegraphics[width=0.95\textwidth]{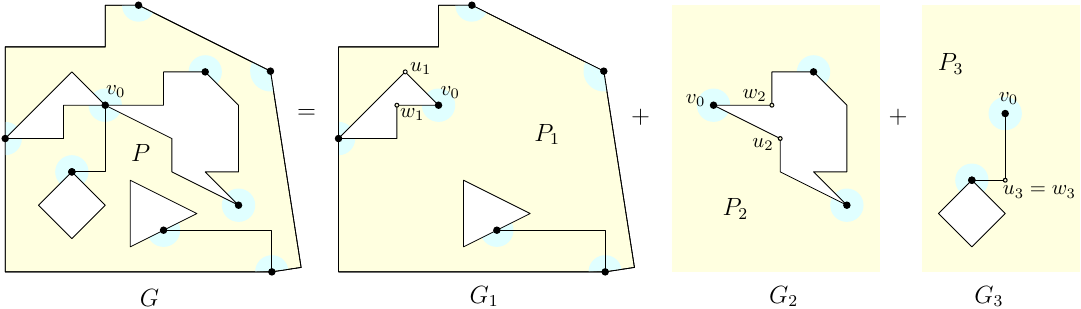}
\caption{A plane ortho-fin multigraph $G$ with a cut vertex $v_0$; and its decomposition into three plane ortho-fin multigraphs $G_1$, $G_2$ and $G_3$. Face $P$ is the intersection of a bounded face $P_1$ of $G_1$ and two unbounded faces $P_2$ and $P_3$ of $G_2$ and $G_3$, respectively.} \label{fig:new}
\end{figure}

\paragraph{General Case.}
It remains to address the general case when $G$ is not necessarily a cycle. We proceed by induction on the number of cut vertices of $G$. In the base case, $G$ does not have any cut vertices, so it is a simple cycle.

For the induction step, assume that $G$ has $c\geq 1$ cut vertices. Let $v_0\in V$ be a cut vertex; see Fig.~\ref{fig:new}. Then $G$ decomposes into $k\geq 2$ maximal multigraphs $G_1,\ldots , G_k$, in which $v_0$ is not a cut vertex. Clearly, $v_0$ is the only common vertex of any two of these multigraphs, and they each have fewer than $c$ cut vertices. The face $P$ of $G$ is contained in some face of every sub-multigraph. Let $P_i$ denote the face of $G_i$ such that $P\subset P_i$ for all $i\in \{1,\ldots ,k\}$. Since $v_0$ is not a cut vertex of $G_i$, then $P_i$ has a unique interior angle incident to $v_0$, $\angle u_i v_0 w_i$ (possibly $u_i=w_i$), which contributes to $\Phi(P_i)$. The exterior angles $\angle w_i v_0 w_i$ are pairwise disjoint, and $P$ has $k$ disjoint angles at $v_0$. Since $P\subset P_i$ for all $i\in \{1,\ldots , k\}$, then all $k$ interior angles of $G$ at $v_0$ are contained in $\angle u_i v_0 w_i$. Using inclusion-exclusion, they sum to $2\pi - \sum_{i=1}^k (2\pi - \angle u_i v_0 w_i) = (\sum_{i=1}^k \angle  u_i v_0 w_i ) - (k-1) 2\pi$. Since $v_0$ is the only common vertex of $G_1,\ldots , G_k$, then $\Phi(P)=(\sum_{i=1}^k \Phi(P_i)) - (k-1) 2\pi$. By induction, $\Phi(P_1),\ldots , \Phi(P_k)$ are multiples of $\pi/2$, consequently $\Phi(P)$ is also a multiple of $\pi/2$. This completes the induction step, hence the entire proof.
\end{proof}

\begin{theorem}\label{thm:orthogonal}
Every plane ortho-fin multigraph on $n\geq 1$ vertices has at most $5n-2$ edges, and this bound is the best possible.
\end{theorem}
\begin{proof}
Let $G=(V,E)$ be a plane ortho-fin multigraph, and denote its faces by $P_1,\ldots , P_f$.
Assume first that $G$ is connected.
Lemma~\ref{lem:angle} implies $\Phi(P_i)\geq \pi/2$ for all $i\in \{1,\ldots , f\}$.
Since the faces of $G$ have pairwise disjoint interiors, then at each vertex $v\in V$, the interior angles at $v$ over all faces sum to $2\pi$.
Summation of the potential over all faces yields
$f\cdot \frac{\pi}{2}\leq \sum_{i=1}^f \Phi(P_i) = 2\pi\cdot n$,
which implies $f\leq 4n$. We combine this inequality with Euler's polyhedron formula, $n-|E|+f=2$
(which holds for connected multigraphs), and obtain $|E|=n+f-2\leq 5n-2$.

It remains to consider the case that $G$ is disconnected. Assume that $G$ has $k$ components with
$n_1,\ldots n_k$ vertices, resp., where $\sum_{i=1}^k n_i=n$.
Each component is an ortho-fin multigraph. Summation of the above bound over all components gives
$|E|=\sum_{i=1}^k (5n_i-2)=5(\sum_{i=1}^k n_i)-2k\leq 5n-2$, as claimed.

\begin{figure}[htbp]
\centering
\includegraphics[width=0.65\textwidth]{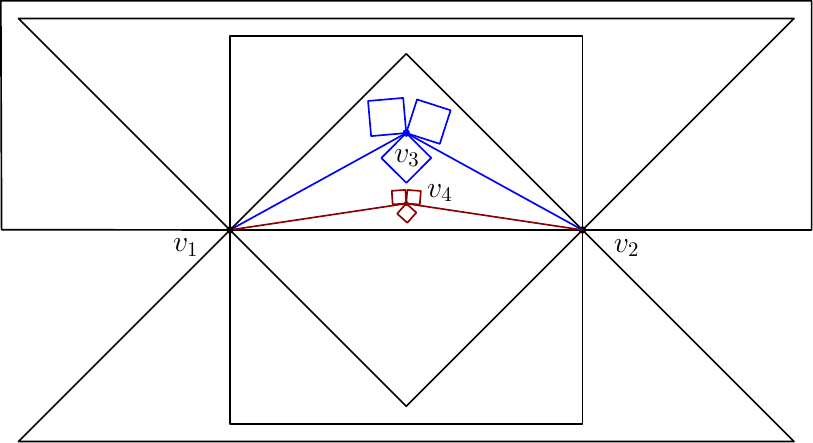}
\caption{Construction for plane ortho-fin multigraphs with $n=2,3,4$ vertices and $5n-2$ edges.} \label{fig:tight}
\end{figure}

For a matching lower bound, we construct plane ortho-fin multigraphs with $n$ vertices and $5n-2$ edges for $n\geq 1$.
It is enough to construct ortho-fin multigraphs in which the potential of every face is precisely $\pi/2$.
For $n=1$, consider a single vertex $v$ and three loops, each of which is a unit square with a corner at $v$.
For $n=2$, let $G_2$ be the orho-fin multigraph on vertex set $\{v_1,v_2\}$ and eight parallel edges as shown in Fig.~\ref{fig:tight}.
For every $n\geq 3$, we construct $G_n$ from $G_{n-1}$ by adding vertex $v_n$ and five new edges, as indicated in Fig.~\ref{fig:tight}:
Two straight-line edges $v_1v_n$ and $v_2v_n$, and three square-shaped loops incident to $v_n$: one loop inside the obtuse triangle $\Delta(v_1 v_2 v_n )$ and two loops outside. By induction, $G_n$ is a plane ortho-fin multigraph with $n$ vertices and $5n-2$ edges.
\end{proof}

\section{Proof of Theorem~\ref{thm:main}}
\label{sec:main}

Let $G=(V,E)$ be a RAC$_2$ drawing with $n\geq 3$ vertices. Assume w.l.o.g.\ that every edge has two bends (by subdividing edge segments if necessary), and the middle segment of every edge has positive or negative slope (not 0 or $\infty$), by rotating the entire drawing by a small angle if necessary. Each edge has two \textbf{end segments} and one \textbf{middle segment}. We classify crossings as \textbf{end-end}, \textbf{end-middle}, and \textbf{middle-middle} based on the crossing segments.

Arikushi et al.~\cite{ArikushiFKMT12} defined a ``block'' on the set of $3|E|$ edge segments. First define a symmetric relation on the edge segments: $s_1\sim s_2$ iff $s_1$ and $s_2$ cross. The transitive closure of this relation is an equivalence relation. A \textbf{block} is the set of segments in an equivalence class. Equivalently, two edge segments, $s_a$ and $s_b$, are in the same block if there exists a sequence of segments $(s_a=s_1, s_2,\ldots, s_t=s_b)$ such that any two consecutive segments cross (necessarily at $90^\circ$ angle).
Note each block consists of segments of exactly two orthogonal directions, and the union of segments in a block is connected; see Fig.~\ref{fig:2}(a) for examples.

\begin{figure}[htbp]
\centering
\includegraphics[width=0.95\textwidth]{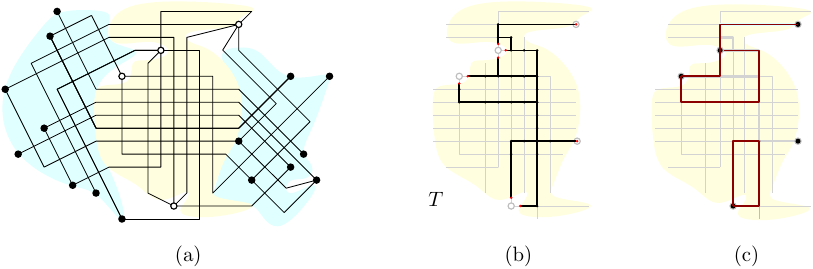}
\caption{(a) Three blocks in a RAC$_2$ drawing.
(b) A spanning tree $T$, after splitting five terminals in $V$ into nine terminals (red dots).
(c) A plane ortho-fin multigraph on the five terminals in $V$.} \label{fig:2}
\end{figure}

\subsection{Matching End Segments in Blocks}
\label{ssec:index2}

Let $B$ be a block of $G=(V,E)$ (refer to Fig.~\ref{fig:2}(a)). Denote by $\mathrm{End}(B)$ the set of end segments in $B$, and let $\mathrm{end}(B)=|\mathrm{End}(B)|$. The segments in $B$ form a connected \textbf{arrangement} $A$, which is a plane straight-line graph: The \emph{vertices} of $A$ are the segment endpoints and all crossings in $B$, and the edges of $A$ are maximal sub-segments between consecutive vertices of $A$. We call a vertex of $A$ a $\textbf{terminal}$ if it is a vertex in $V$ (that is, an endpoint of some edge in $E$). If a terminal $p$ is incident to $k>1$ edges of the arrangement $A$, we shorten these edges in a sufficiently small $\varepsilon$-neighborhood of $p$, and split $p$ into $k$ terminals (Fig.~\ref{fig:2}(b)). We may now assume that each terminal has degree 1 in $A$, hence there are $\mathrm{end}(B)$ terminals in $A$.

Let $T$ be a minimum tree in $A$ that spans all terminals. It is well known that one can find $\lfloor \frac12 \mathrm{end}(B)\rfloor$ pairs of terminals such that the (unique) paths between these pairs in $T$ are pairwise edge-disjoint (e.g., take a minimum-weight matching of $\lfloor \frac12 \mathrm{end}(B)\rfloor$ pairs of  terminals). If $k>1$ terminals correspond to the same vertex $p\in V$, then we can extend the paths by $\varepsilon>0$ to $p$, and the extended paths are still edge-disjoint.
Let $\mathcal{E}(B)$ be the set of these paths (Fig.~\ref{fig:2}(c)).

Let $\mathcal{E}$ be the union of the sets $\mathcal{E}(B)$ over all blocks $B$; see Fig.~\ref{fig:3}(left).
Since every path in $\mathcal{E}$ is a simple polygonal path between points in $V$, it can be interpreted as the drawing of an edge in a multigraph on the vertex set $V$, and so $\mathcal{E}$ is a set of edges on $V$.
With this interpretation, $H=(V,\mathcal{E})$ is a plane ortho-fin multigraph. Indeed, the edges of $H$ are paths in $\mathcal{E}$. This means that no two edges of $H$ cross. Each edge of $H$ is a path within the same block, and so the first and last segment of each edge of $H$ are either parallel or orthogonal to each other.

We say that an edge $e\in \mathcal{E}$ \textbf{represents} an edge $f\in E$ if the first or last edge segment of $f$ contains the first or last edge segment of the edge $e$. By definition, each edge $e\in \mathcal{E}$ represents at most two edges in $E$.

\begin{figure}[htbp]
\centering
\includegraphics[width=\textwidth]{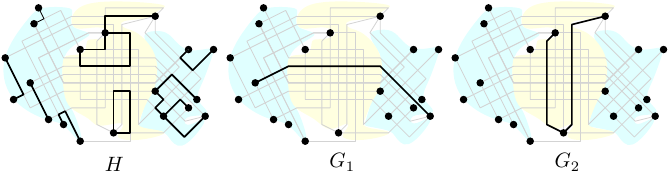}
\caption{Graphs $H$ (left), $G_1$ (middle) and $G_2$ (right) for the RAC$_2$ drawing in Fig.~\ref{fig:2}.
The original RAC$_2$ drawing is shown in light gray for comparison.} \label{fig:3}
\end{figure}

By Theorem~\ref{thm:orthogonal}, the graph $H=(V,\mathcal{E})$ has at most $5n-2$ edges, and so it represents at most $2(5n-2)=10n-4$ edges of $G$.

\subsection{Gap Planar Graphs}

Let $E_0\subset E$ be the set of edges in $G=(V,E)$ that are not represented in $H$, and let
$G_0=(V,E_0)$. Clearly $G_0$ is a RAC$_2$ drawing with $n$ vertices.

\begin{lemma}\label{lem:nonrep}
In the drawing $G_0=(V,E_0)$, there is no end-end crossing, and each middle segment is crossed by at most one end segment.
\end{lemma}
\begin{proof}
A block of $G_0$ is a subarrangement of a block of $G$. In every block of $G$, there is at most one end segment whose edge is not represented by some edge in $H=(V,\mathcal{E})$. Consequently, in every block of $G_0$, there is at most one end segment. Both claims follow.
\end{proof}

Partition $G_0=(V,E_0)$ into two subgraphs, denoted $G_1=(V,E_1)$ and $G_2=(V,E_2)$, such that $E_1$ contains all edges in $E$ whose middle segments have negative slopes, and $E_2=E_0\setminus E_1$; see Fig.~\ref{fig:3} for an example.

\begin{lemma}\label{lem:slopes}
In each of $G_1=(V,E_1)$ and $G_2=(V,E_2)$, all crossings are end-middle crossings, and every middle segment has at most one crossing.
\end{lemma}
\begin{proof}
Since $G_1$ (resp., $G_2$) is a RAC$_2$ drawing, where all middle segments have positive (resp., negative) slopes, then the middle segments do not cross. Combined with Lemma~\ref{lem:nonrep}, this implies that all crossings are end-middle crossings, and every middle segment crosses at most one end segment.
\end{proof}

\begin{figure}[htbp]
\centering
\includegraphics[width=0.5\textwidth]{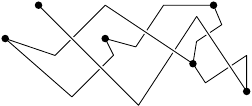}
\caption{A RAC$_2$ drawing: All crossings are between end- and middle-segments, and every middle-segment has positive slope and at most one crossing.} \label{fig:4}
\end{figure}

Bae~et al.~\cite{BaeBCEE0HKMRT18} defined a \textbf{$k$-gap planar} graph, for an integer $k\geq 0$, as a graph $G$ that can be drawn in the plane such that (1) exactly two edges of $G$ cross in any point, (2) each crossing point is \emph{assigned} to one of its two crossing edges, and (3) each edge is assigned with at most $k$ of its crossings.

\begin{lemma}\label{lem:gap}
Both $G_1=(V,E_1)$ and $G_2=(V,E_2)$ are 1-gap planar.
\end{lemma}
\begin{proof}
Every crossing is an end-middle crossing by Lemma~\ref{lem:slopes}. Assign each crossing to the edge that contains the middle segment involved in the crossing. Then each edge is assigned with at most one crossing by Lemma~\ref{lem:slopes}; see Fig.~\ref{fig:4}.
\end{proof}

Bae~et al.~\cite{BaeBCEE0HKMRT18} proved that every 1-gap planar graph on $n\geq 3$ vertices has at most $5n-10$ edges, and this bound is the best possible for $n\geq 5$. They have further proved that a multigraph with $n\geq 3$ vertices that has a 1-gap planar drawing in which no two parallel edges are homotopic has at most $5n-10$ edges.
It follows that $G_1$ and $G_2$ each have at most $5n-10$ edges if $n\geq 3$.

\medskip\noindent\textbf{Proof of Theorem~\ref{thm:main}:}
Let $G=(V,E)$ be a RAC$_2$ drawing. Graph $H=(V,\mathcal{E})$ represents at most $2(5n-2)=10n-4$ edges of $G$ by Theorem~\ref{thm:orthogonal}. The remaining edges of $G$ are partitioned between $G_1$ and $G_2$, each containing at most $5n-10$ edges for $n\geq 3$; see Fig.~\ref{fig:4}. Overall, $G$ has at most $20n-24$ edges if $n\geq 3$.
\hfill\qed


\section*{Acknowledgements}
Work on this paper was initiated at the Tenth Annual Workshop on Geometry and Graphs held at the Bellairs Research Institute,  February 3--10, 2023.
The author is grateful to Michael Kaufmann, G\"unter Rote, and Torsten Ueckerdt for stimulating discussions;
and also thanks the reviewers of GD~2023 and JGAA for many helpful comments and suggestions.


\end{document}